\documentclass[11pt]{article}

\usepackage{amsmath, amssymb, amsthm, amsfonts, latexsym, graphicx,
  color, booktabs}

\usepackage{url}
\usepackage{algorithmic}
\usepackage{algorithm}
\usepackage{verbatim}
\usepackage{tikz}
\usetikzlibrary{arrows, shapes,positioning}
\usepackage{tkz-graph}
\usepackage{pstricks}
\usepackage[tiling]{pst-fill}
\usepackage{multido} 

\usepackage{fullpage}
\usepackage{amsmath,amsfonts,amsthm,mathrsfs,mathpazo,xspace,graphicx}
\usepackage{endnotes,needspace}
\usepackage{color}
\usepackage{bm}
\usepackage{times}
\usepackage{enumerate}
\usepackage{amssymb,latexsym}

\newcommand{\bra}[1]{\langle#1|}

\newcommand{\ip}[1]{\langle {#1} \rangle}

\newcommand{\norm}[1]{\| #1 \|}

\DeclareMathOperator{\tr}{tr}

\DeclareMathOperator{\poly}{poly}

\newcommand{\beq}{\begin{equation}}
\newcommand{\eeq}{\end{equation}}
\newcommand{\beqn}{\begin{equation*}}
\newcommand{\eeqn}{\end{equation*}}

\newcommand{\C}{\ensuremath{\mathbb{C}}}

\newcommand{\R}{\ensuremath{\mathbb{R}}}

\newcommand{\cE}{\ensuremath{\mathcal{E}}}

\newtheoremstyle{promise}
{3pt}% Space above
{3pt}% Space below
{}% Body font
{}% Indent amount
{\itshape}% Theorem head font
{:}% Punctuation after theorem head
{.5em}% Space after theorem head
{}% Theorem head spec (can be left empty, meaning `normal')

\theoremstyle{promise}

\newcommand{\be}{\begin{eqnarray}}
\newcommand{\ee}{\end{eqnarray}}

\newcommand{\eps}{\varepsilon}

  \theoremstyle{plain}
  \newtheorem{theorem}{Theorem}
  \newtheorem{lemma}[theorem]{Lemma}

  \newtheorem{claim}[theorem]{Claim}
  
  \theoremstyle{definition}
  
  \newtheorem{definition}{Definition}

  \theoremstyle{plain}
  
\newcommand{\ii}{\mathbb{I}}

\newcommand{\oo}[1]{\Theta\left(#1\right)} % on the order of

\newcommand{\ket}[1]{\vert #1 \rangle}

\newcommand{\braket}[2]{\langle #1 \vert #2 \rangle}

\newcommand{\nc}{\newcommand}
\nc{\bbC}{\mathbb{C}}
\nc{\bbI}{\mathbb{I}}
\nc{\bitem}{\begin{itemize}}
\nc{\eitem}{\end{itemize}}
\nc\benum{\begin{enumerate}}
\nc\eenum{\end{enumerate}}

\newcommand{\secref}[1]{Section~\ref{sec:#1}}

\newcommand{\lemref}[1]{Lemma~\ref{lem:#1}}
\newcommand{\thmref}[1]{Theorem~\ref{thm:#1}}
\newcommand{\clmref}[1]{Claim~\ref{clm:#1}}

\newcommand{\tabref}[1]{Table~\ref{tab:#1}}

\def\cE{\mathcal{E}}

\def\cH{\mathcal{H}}

\newcommand{\proj}[1]{\ket{#1}\bra{#1}}

\def\be#1\ee{\begin{equation}#1\end{equation}}
\def\bea#1\eea{\begin{eqnarray}#1\end{eqnarray}}
\def\beas#1\eeas{\begin{eqnarray*}#1\end{eqnarray*}}
\def\ba#1\ea{\begin{align}#1\end{align}}
\def\bas#1\eas{\begin{align*}#1\end{align*}}
\def\bpm#1\epm{\begin{pmatrix}#1\end{pmatrix}}

\def\eq#1{(\ref{eq:#1})}

\def\L{\left} 
\def\R{\right}

\def\ot{\otimes}

\usepackage[margin=1in]{geometry}

\newcommand{\ignore}[1]{}
\def\begsub#1#2\endsub{\begin{subequations}\label{eq:#1}\begin{align}#2\end{align}\end{subequations}}
\nc\mnb[1]{\medskip\noindent{\bf #1}}
\nc\qand{\qquad\text{and}\qquad}

\begin{document}

\title{\Large \textbf{Local tests of global entanglement \\and a counterexample to the generalized area law}}
\author{Dorit Aharonov, Aram W. Harrow, Zeph Landau, Daniel Nagaj,
Mario Szegedy, Umesh Vazirani}

\maketitle
\vspace{-5mm}

\begin{abstract}
We introduce a technique for applying quantum expanders in a distributed fashion, and use it to 
solve two basic questions: testing whether a bipartite quantum state shared by two parties is
the maximally entangled state and disproving a generalized area law. In the process these two 
questions which appear completely unrelated turn out to be two sides of the same coin. Strikingly
in both cases a constant amount of resources are used to verify a global property.
\end{abstract}

\section{Introduction}

In this paper we address two basic questions:

\begin{enumerate}
\item Can Alice and Bob test whether their joint state is maximally entangled while exchanging only a constant number of qubits?
More precisely, Alice and Bob hold two halves of a quantum state 
$\ket{\psi}$ on a $D^2$-dimensional space for large $D$, and would like to check whether $\ket{\psi}$ is the maximally entangled state
$\ket{\phi_D} = \frac{1}{\sqrt{D}} \sum_x \ket{x}\ket{x}$ or whether it is orthogonal to that state. 
The first entanglement tester  
is the {\em hashing protocol} of the influential 1996 paper by Bennett, DiVincenzo, Smolin and Wootters~\cite{BDSW96}; further results
are summarized in table \tabref{ent-test}. Entanglement testing has found various applications, including entanglement distillation and error correction~\cite{BDSW96}, state authentication~\cite{BCGST02} and bounding the communication capacities of 
bipartite unitary operators~\cite{HL07}. 
As can be seen from this table, 
all known protocols for this task 
require resources (communication, shared randomness or catalyst) 
which grow with $D$ \cite{BDSW96, BCGST02, HL07}.  

\item Is there a counterexample to the generalized area law? A sweeping conjecture in condensed matter physics, and one of the most
important open questions in quantum Hamiltonian complexity theory, is
the so called ``area law,'' which asserts that ground states of quantum
many body systems on a lattice 
have limited entanglement. Specifically, assume the system is described by a gapped local Hamiltonian\footnote{Here and later, by gapped local Hamiltonian we mean a Hamiltonian whose difference between ground energy and next excited 
energy is $\Omega(1)$.} $H = H_1 + \ldots + H_m$, where each $H_i$
describes a local interaction between two neighboring particles of a
lattice. The area law conjectures that for every subset $S$ of the
particles, the entanglement entropy between $S$ and $\bar{S}$ for the
ground state of $H$ is bounded by a constant times the size of the
boundary of $S$.  The area law, which has been proven for 1-D
lattices~\cite{ref:Has07} and is conjectured for higher degree
lattices, is of central importance in condensed matter physics as it
provides the basic reason to hope that ground states of gapped local
Hamiltonians on lattices  might have a (relatively) succinct classical
description. The {\em generalized} area law (a folklore conjecture) 
transitions from this physically motivated phenomenon to a very clean and general graph theoretic formulation, where in place of edges of the lattice, the terms of the Hamiltonian correspond to edges of an arbitrary graph.  The generalized area law then states that for any subset $S$ of vertices (particles), the entanglement entropy between $S$ and $\bar{S}$ for the ground state is bounded by some constant times the cut-set of $S$ (the number of edges leaving $S$).
\end{enumerate}

We affirmatively answer both questions, based on a common technique that may be thought of as applying quantum expanders in a distributed fashion.  Indeed these two questions which at first sight seem completely unrelated turn out to be two sides of the same coin.

\begin{table}
\begin{tabular}{@{}llp{1.5in}p{3in}@{}}
\toprule
reference & form of $\ket\psi$ & communication  required & other resources \\
\midrule
\cite{BDSW96} & $n$ EPR pairs & $O(n\log(1/\eps))$ & consumes only $O(\log(1/\eps))$ EPR pairs \\
\cite{BCGST02} & $n$ EPR pairs & $(2+o(1))\log(n/\eps)$ & $O(n/\log(n/\eps))$ bits of shared randomness \\
\cite{HL07} & $n$ EPR pairs & $O(\log(1/\eps))$ & $n/\eps$ EPR pairs \\
this paper & $n$ EPR pairs & $O(\log(1/\eps))$ & \\
\bottomrule
\end{tabular}
\caption{Comparison of different entanglement-testing protocols.  When
  we say that communication $x$ is required, this means that we need
  to consume either $x$ qubits or, alternatively (by
  teleportation), $2x$ classical bits and $x$ EPR pairs.   The exception is the first row, which 
uses classical communication to verify entanglement, and hence the
entanglement cost is lower. \label{tab:ent-test}} 
\end{table}

\subsection{Main idea and results}

The main ingredient in both proofs is the notion of quantum expanders, which we discuss further in \secref{expanders}.  
A quantum expander can be thought of as a collection of $d$ unitaries
$U_1,\ldots,U_d$, (think of $d$ as a constant) each acting on a (possibly large)
dimension-$D$ Hilbert space.  For any 
matrix $X$ on the $D$-dimensional Hilbert space, the operator associated with the expander, 
$\mathcal{E}(X)= \frac{1}{d}\sum_{i=1}^d U_i X U_i^\dagger$,
has the unique eigenvalue 1 for the eigenvector $X=\bbI$ and next highest singular value $\lambda < 1$. It thus shrinks any
matrix orthogonal to the identity by a constant factor.
The key to the results in the paper is an equivalent way to view quantum expanders, by considering their action on maximally entangled states.  It is well known that for any $U$, $U\otimes U^*$ acting on the maximally entangled state $\ket{\phi_D} = \frac{1}{\sqrt D}\sum_{i=1}^D \ket i \ot \ket i$ leaves it as is.  Of course, this remains true even if $U$ is drawn uniformly at random from the set $U_1,\ldots,U_d$ of the expander.  Remarkably, even though quantum expanders use only a constant number $d$ of unitaries, they leave intact {\it only} the maximally entangled state, and cause all other states to shrink by at least a constant amount.

For the entanglement-testing problem, we use the above intuition to
derive a communication protocol which uses only a constant number of
qubits, and detects a maximally entangled state of arbitrary
dimension.  This is described in Section \ref{sec:EPR-test}. 
The idea is to determine whether Alice and Bob share a
state that is invariant under $U_i \ot U_i^*$ for $i=1,\ldots,d$, or
far from invariant; i.e. whether the shared state is $\ket{\phi_D}$ or
something orthogonal.  To achieve this with $O(1)$ communication,
suppose Alice and Bob each had access to a joint register initialized
with $1/\sqrt{d} \sum_{i=1}^d \ket{i}$.  Each could then apply controlled
operators from this register to their share of the state $\ket{\psi}$:
Alice would apply a controlled $U_i$ and Bob a controlled $U_i^*$.
This ``shared register'' model could more naturally be achieved by
having Alice create the state and perform her controlled $U_i$ before
sending the state to Bob, who then performs his controlled $U_i^*$.
Bob should then test that the control state remained intact, which happens iff the original
state of the $D$-dimensional registers was indeed the maximally
entangled state. With only a little more algebra, this shows that for
any $D$, $\epsilon > 0$, there exists a protocol which uses $O(\log
1/\epsilon)$ qubits of communication, after which Bob always accepts
if the shared state is $\ket{\phi_D}$. If the shared state is
orthogonal to $\ket{\phi_D}$, he accepts with probability at most
$\epsilon$. Moreover, if Alice and Bob do start with the maximally
entangled state $\ket{\phi_D}$, the protocol does not damage the
state.

\begin{figure}[h]
\begin{center}
\includegraphics[width=9cm]{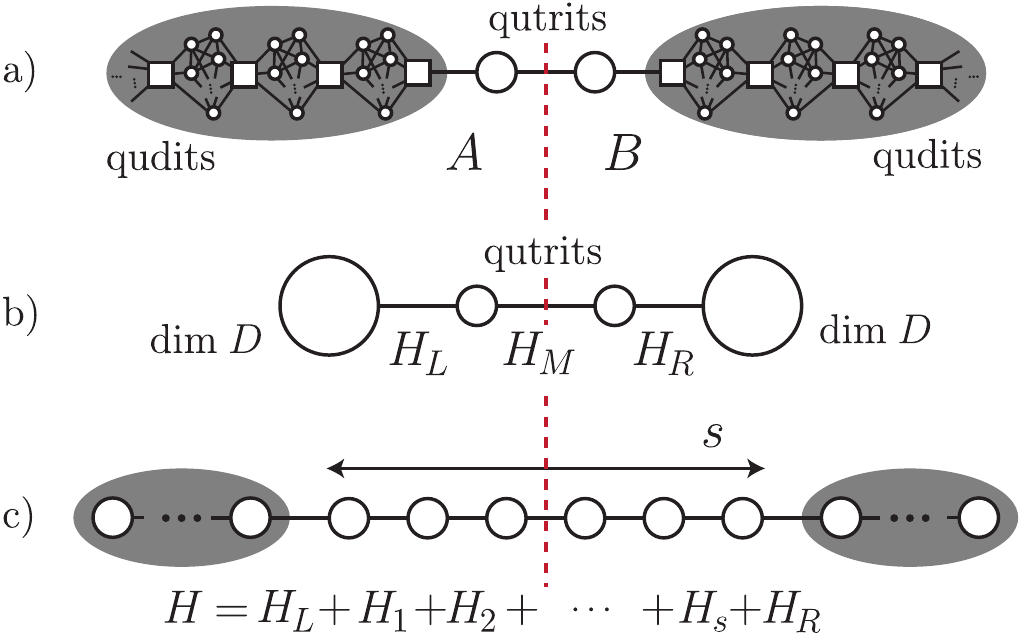}
\caption{\small a) A counterexample to the generalized area law, consisting of a chain of complete graphs separated by the middle edge. The entropy across the cut is shown in \secref{concreteH} to
grow as $\Omega\left(n^{c}\right)$, where $n$ is the total number of particles. b) A four-particle construction, analyzed in \secref{HLMR}. c) Short-chain framework for proving 1-D area law, from \cite{AradKLV12area}.}
\label{fig-intro}
\end{center}
\end{figure}

For a counterexample to the generalized area law, we use the above
intuition to exhibit a gapped local Hamiltonian acting on the graph
featured in Figure~\ref{fig-intro}\,a), for which the entanglement
entropy of the ground state across the middle cut is $\Omega (n^c)$
for some $0 < c < 1$ (rather than $O(1)$ as predicted by the
generalized area law). The core step in generating this example is the
construction of a simpler system consisting of four particles on a
line in Figure~\ref{fig-intro}\,b): two particles of dimension $d = 3$
(qutrits) in the middle, and two particles of dimension $D$ at the two
ends, with arbitrarily large $D$. The gapped Hamiltonian is of the
form $H = H_L + H_M + H_R$, where $H_L$ acts between the left particle
and the left qutrit, $H_M$ between the two qutrits, and $H_R$ between
the right qutrit and the right particle. Crucially, the entanglement
entropy of the ground state across the middle cut is $\Omega(\log D)$,
as shown in Section~\ref{sec:HLMR}.  The idea here, like in the
communication protocol, is to use the middle particles to synchronize
the application of a quantum expander on the left and right sides.
This requires only a single term of the Hamiltonian, acting on two
$d$-dimensional particles.

Enforcing a large amount of entanglement (in the ground state) by the
single constraint $H_M$ acting on a constant-dimensional system is a
surprising quantum phenomenon.  In the analogous probabilistic
situation, consider a graph whise vertices are each associated with
constant-dimensional variables, and whose edges are associated with
classical constraints. Each constraint forbids some subset of the
possible assignments to the variables at the two ends of the
edge. This describes a constraint satisfaction problem
(CSP)\footnote{This analogy between local Hamiltonians and constraint
  satisfaction problems is commonly used in quantum Hamiltonian
  complexity, see e.g., \cite{AAV13}.}.  Now consider the uniform
distribution over the set of all possible solutions to this set of
constraints, namely all assignments that violate no constraint.  It is
easy to see that the middle constraint in the graph in
Figure~\ref{fig-intro}\,a) can only enforce a convex combination of a
constant number of product distributions \footnote{This phenomena can
  be viewed as the zero temprature case; it in fact extends also to
  the Gibbs distribution at any temprature, where the two endpoints of
  a chain are always conditionally independent given the values of the
  spins in the middle.}.

This simple example of a four-particle system is already important within the context of proofs of the 1-D area law and prospects for extending those techniques to higher dimensions. 
The best current 1-D area law \cite{AradKLV12area} works 
within a model very similar to our four-body Hamiltonian, 
except the middle link in \cite{AradKLV12area} is extended into a finite chain of $s = \Omega(\log^2 d/\epsilon)$ particles, each of dimension $d$ (see figure~\ref{fig-intro}\,c). 
This yields an area law bound of  $S_{1-D} = O(\log^3(d) / \epsilon)$ across the middle cut. It was observed in \cite{AradKLV12area} that
any slight improvement in the exponent of $\log d$ would imply a non-trivial sub-volume law for 2-D systems. The crucial parameter in improving the result is the length of the middle chain;  
Our four-body Hamiltonian shows that in the extreme case of a length $1$ chain,
no area law holds. 
Understanding the intermediate regime is thus an important open question. 

Our four-particle example involves non-physical particles of arbitrarily large dimension. 
In Section~\ref{sec:concreteH} this example is converted to a counterexample 
to the generalized area law with  
bounded dimensional particles (albeit with unbounded degree of interaction). 
This is done by applying Kitaev's circuit-to-Hamiltonian construction to implement the $U_i$, 
followed by an application of the strengthening gadgets of \cite{NagajCao} 
(see  Section~\ref{sec:concreteH} for details). 

What is the connection between our two results? 
The above described Hamiltonian constructions are based on quantum expanders, just like our entanglement testers. 
In Section~\ref{sec:comm} we explore a deeper connection between very efficient communication protocols for EPR testing and violations of generalized area laws. We show that it is possible to derive a counterexample to the generalized area law, starting from a solution to the first problem (an EPR testing protocol with limited communication) and converting it using Kitaev's circuit-to-Hamiltonian construction into a Hamiltonian violating the generalized area law.  This, we believe, points at a fundamental link
between the two seemingly unrelated topics. We discuss this and many other open questions and related work 
in Section \ref{sec:discussion}.

{\em Notation:} For a matrix $X$, let $X^*$ be the entry-wise complex conjugate of $X$ and $X^\dag$ the transpose of $X^*$.  Define the Frobenius norm $|X|:=\sqrt{\tr X^\dag X}$; the operator norm $\|X\|$ is the largest singular value of $X$.

\section{Quantum Expanders}\label{sec:expanders}
The key structural component to our results are quantum expanders.  
We will only make use here of expanders based on applying one out of $d$ 
unitaries at random  
(a more general definition using Kraus operators exists).   

\begin{definition}
The operator $\mathcal{E}: L(\C ^D) \rightarrow L(\C^D)$ (here we use $L(\C^D)$ to denote the set of linear operators on $\C^D$) is termed a {\it $(D, d, \lambda)$ quantum expander} if 
\bitem
\item There are $d$ unitaries, $U_1, U_2,...,U_d$,
such that $\mathcal{E}(X) = \frac{1}{d} \sum_{i=1}^dU_i X U_i^\dagger $. 
\item Interpreted as a linear map, 
$\cE$ has second-largest singular
 value $\leq \lambda$.
\eitem
\end{definition}

Just as classical expanders may be thought of as constant-degree 
approximations to the complete graph, quantum expanders are 
constant-degree approximations to the application of unitaries drawn 
at random from the
Haar measure. 

By definition, the identity map $X=\mathbb{I}$ 
is the unique fixed point of $\mathcal{E}$.  
The second condition is equivalent to saying 
that for any $X$ with $\tr(X)=0$
\be | \mathcal{E} (X)| \leq \lambda |X|
\label{eq:frob-exp}.\ee
This interpretation suggests an alternate formulation 
where we think of each $X\in L(\C^D)$ as  a vector in $\C^D \otimes \C^D$ and the operator $\mathcal{E}$ then gets mapped to the 
operator $\hat{\mathcal{E}} = \frac{1}{d} \sum_{i=1}^dU_i \otimes U_i^* $.  Then $\hat{\mathcal{E}}$ 
fixes the maximally entangled state $\ket{\phi_D} =\frac{1}{\sqrt{D}}\sum_{x\in[D]} \ket{x}\ket{x}$, and has second largest singular value $\lambda$.  

Quantum expanders were introduced independently in
\cite{Hastings-expander1} and \cite{BST07-expander-orig} although many of the relevant ideas
were implicit in \cite{AS04}. In \cite{Hastings-expander2}, it was proved that 
taking $U_i$ for $i\in \{1,...,d\}$ 
to be Haar uniform results in a ``Ramanujan'' expander
with high probability; that is, $\lambda \approx 1/\sqrt{d}$.
Since random unitaries cannot be constructed
efficiently, other work~\cite{BST08-expander, H-cayley,
QuantumMargulisExpanders, HH09} 
gave efficient constructions, in which the unitaries can be applied 
by polynomial-size quantum circuits. 
Essentially all of these constructions achieve $\log(d) = O(\log
1/\lambda)$.
For our communication protocols, we will need $d$ to be a variable 
(since the error depends on it); 
whereas for the area law counter example, we will take $d$ to be a
small constant.  In what follows we will assume for simplicity of exposition that $d=3$ is possible,
although the smallest $d$ that has been verified is $d=8$ using
\cite{QuantumMargulisExpanders}.

Why are expanders relevant to our results? To understand the gap condition better, 
let us see why 
$\ket{\phi_D}$, the maximally entangled state on $\bbC^D \ot \bbC^D$,
 is a $+1$ eigenvector. 
Observe that for any $D\times D$ matrix $X$, 
we have $(X\otimes I)
\ket{\phi_D} =(I\otimes X^T) \ket{\phi_D}$.
Thus $\frac{1}{d}\sum_{i=1}^d (U_i \ot U_i^*)\ket{\phi_D} = 
\frac{1}{d}\sum_{i=1}^d (U_iU_i^\dag \ot I)\ket{\phi_D} = \ket{\phi_D}$.
Since the second-largest singular 
value of $\hat{\cE}$ is $\lambda$, then we have
\be \L\|\hat{\cE} - \proj{\phi_D} \R\| = \lambda \label{eq:exp-gap}.\ee

Thus, $\hat{\cE}$ gives an approximation of a projector onto $\ket{\phi_D}$
up to operator-norm error $\lambda$. 

In the rest of our paper we will explore two settings in which this
allows us to use resources proportional to $d$ (which we should think
of as small) to force a state on $\bbC^D \ot \bbC^D$ (with $D$ large)
to be close to $\ket{\phi_D}$.
\bitem
\item In \secref{EPR-test} we will show how $\log(d)$ qubits of communication
  can perform the projective measurement $\{\proj{\phi_D}, I -
  \proj{\phi_D}\}$ up to error $1/d^{\Omega(1)}$.
\item In \secref{HLMR} we will show how interactions between a pair of
  constant-dimensional 
  and a pair of two $D$-dimensional particles can have a ground state
  with maximal entanglement on the $D$-dimensional particles and a constant gap.
\eitem

\section{A communication protocol for certifying global entanglement}
\label{sec:EPR-test}

\subsection{The EPR testing problem}
As above, set $\ket{\phi_D}$ to be 
the maximally entangled state on $\bbC^D \ot \bbC^D$.  The EPR testing
problem is to determine whether a given shared state
$\ket\psi\in\bbC^D \ot \bbC^D$ is equal to or orthogonal to
$\ket{\phi_D}$.   More precisely, two parties (Alice and Bob) would
like to simulate the joint two-outcome POVM $\{\proj{\phi_D}, \bbI - \proj{\phi_D}\}$.

An $(D,\eps)$ EPR tester is a communication protocol for performing a two-outcome measurement $\{M,\bbI-M\}$ such that $\|M - \proj{\phi_D}\|\leq \eps$.   

In general EPR testers may differ in a variety of ways:
\bitem
\item If $M\geq \proj{\phi_D}$ then we say the EPR tester has {\em perfect completeness}.
\item The communication requirements and computational efficiency may vary.
\item The protocol may be performed with quantum or classical communication.  If quantum communication is used, then it is reasonable to assume that upon input $\rho$ the post-measurement state is $M^{1/2}\rho M^{1/2} / \tr[M\rho]$ or 
$(\bbI-M)^{1/2}\rho (\bbI-M)^{1/2} / \tr[(\bbI-M)\rho]$, depending on the
outcome.  If classical communication is used, we need to also consume
some entanglement.  We say that the test consumes $k$ EPR pairs if given 
an input of $n$ EPR pairs, it outputs at least $n-k$ EPR pairs (up to
$\eps$ error) when it
reports success. There are no guarantees for orthogonal input states.
\eitem

We are aware of three previous implementations of EPR testers
(previously described in \tabref{ent-test}).  Ref.~\cite{BDSW96} gave a
$\left(2^n,2^{-k}\right)$ EPR tester with perfect completeness that used a
message of $O(nk)$ classical bits and consumed $k$ EPR pairs.  This
was improved by \cite{BCGST02} to a $(2^n,\frac{2n}{k(2^k+1)})$ EPR
tester that sent $2k$ classical bits, consumed $k$ EPR pairs and used
$\approx n/k$ bits of shared randomness.   The paper \cite{HL07}
provides a protocol which uses only $\log 1/\eps$ communication qubits, but with the
assistance of an additional $n/\eps$ trusted EPR pairs.   Our
protocol achieves a protocol with this amount of communication,
without the need for any extra resources. 

\subsection{EPR testing with constant quantum communication}
\label{sec:EPRprotocol}

Our main result in this
section is an EPR tester using only a constant amount of quantum communication that is {\em
  independent} of the dimension $D$. 

\begin{theorem}\label{thm:EPR-test}
For any $D$ and any $\eps>0$ there exists a $(D,\eps)$ EPR tester with
perfect completeness using one-way communication from Alice to Bob.
The protocol has several variants:
\bitem
\item Using $(2 + o(1))\log(1/\eps)$ qubits sent from Alice to Bob, but $\poly(D)$ run-time.
\item Using $C\log(1/\eps)$ qubits sent from Alice to Bob and
  $\poly\log(D)$ run-time for some universal constant $C>0$.
\item Using either $(8 + o(1))\log(1/\eps)$ or $\approx\!
  4C\log(1/\eps)$ classical bits sent from Alice to Bob (depending on
  whether computational efficiency is needed) and consuming the same
  number of EPR pairs. 
\item Using 2 bits of shared randomness and 1 qubit of communication,
  $\poly\log(D)$ run-time and achieving $\eps=\frac{8 + \sqrt 5}{16} \leq 0.64$.
\eitem
\end{theorem}

We remark that replacing the state $\ket{\phi_D}$ in \thmref{EPR-test}
with a general entangled state can result in a much larger (and
$D$-dependent) communication cost~\cite{HL07}.  Thus we refer to the
result as an EPR tester rather than a general entanglement tester.

One application of this result relates to the open question of whether entanglement helps
quantum communication complexity.  Classically, shared randomness 
does not significantly reduce communication
complexity because large random strings can be replaced by
pseudo-random strings that fool protocols\cite{Newman91}.  This is
called a blackbox reduction because it replaces the random input but
does not change the protocol.  Quantumly such blackbox
reductions are ruled out by efficient entanglement-testing
protocols, since they cannot be fooled by any low-entanglement state.
A similar result is in \cite{JRS08} but their construction does not
yield an EPR tester.  See also \cite{SZ08} for a non-blackbox entanglement
reduction that increases the communication cost by an exponential amount.

\begin{proof}[Proof of \thmref{EPR-test}]

The main idea is to interpret the results of \secref{expanders} as a way to 
test maximally entangled states. By \secref{expanders} it
 suffices for Alice and Bob to
implement a two-outcome measurement $\{M,\bbI-M\}$ on their shared state
with $M = \hat\cE$ for $\cE$ a $(D,k,\sqrt\eps)$ expander. 
However, it is not immediately clear how to implement this
measurement. To do this, we will use a trick that has been 
used in a variety of contexts (e.g. \cite{BBDEJM96}, \cite{HL07} and Section 2.2.2 of
\cite{money-ICS}) and can be thought of as a variant of phase estimation.
The protocol (depicted in Figure~\ref{fig:commgames}a) is as follows.
\benum
\item Alice and Bob initially share a state in registers $L,R$.
\item Alice prepares the $\log(d)$-qubit state $\frac{1}{\sqrt{d}}\sum_{i=1}^d
  \ket i$ in register $a$.
\item She performs $W=\sum_{i=1}^{d} \proj i \ot U_i$ on
$a,L$.
\item She sends system $a$ to Bob.
\item Bob performs $W^* = \sum_{i=1}^d \proj i \ot U_i^*$ on
$a,R$.
\item Bob does a two-outcome measurement on $a$, with the “accept”
  outcome corresponding to the state $\frac{1}{\sqrt{d}} \sum_{i=1}^d \ket i$
  and the ``reject'' outcome corresponding to the orthogonal subspace.
\eenum

If Alice and Bob start with the shared state $\ket\psi$, then step 5,
their state is
\be 
\frac{1}{\sqrt d}\sum_{i=1}^d \ket{i}^a \ot (U_i \ot U_i^*)
\ket{\psi}^{LR}.\ee
Step 6 then accepts with probability equal to the norm squared of $\frac{1}{d}\sum_{i=1}^d (U_i \ot U_i^*)
\ket{\psi} =: M\ket\psi$ where we have defined 
\be M = \frac{1}{d}\sum_{i=1}^d U_i \ot U_i^*.\ee
This results in the two-outcome measurement $\{M^\dag M, I - M^\dag M\}$.
By Eq.~\eq{exp-gap},
$M^\dag M$ is $\eps$ close to the desired measurement operator
$\proj{\phi_D}$.  

The communication cost is $\log(d)$.  If we do not care about
computational efficiency, we can obtain $\eps \approx 1/\sqrt{d}$
using random unitaries~\cite{Hastings-expander2}.  For a
$\poly\log(d,D)$ run-time, we can iterate an efficient expander;
e.g. applying the construction of \cite{QuantumMargulisExpanders} $k$
times yields $d=8^k$ and $\eps \leq (2\sqrt 5/8)^k$.  To use classical
bits instead, we first use the construction of \cite{BCGST02} which
uses $O(\log 1/\eps)$ classical bits to verify $O(\log 1/\eps)$ EPR
pairs. Those EPR pairs are then used to teleport the qubits in the
above protocol, which can therefore be applied to verify the rest of
the EPR pairs.

For the last construction that uses two rbits and one qubit, we start
with the quantum Margulis expander~\cite{QuantumMargulisExpanders}.
This consists of unitaries $U_1,\ldots,U_4, U_5 = U_1^\dag, \ldots,
U_8=U_4^\dag$.  The modified protocol is as follows. Let $r \in
\{1,2,3,4\}$ be the value of the shared randomness.  Run the above
protocol with the pair of unitaries $\{I, U_r\}$.  The resulting
measurement operator, conditioned on $r$, is $M_r = \frac{I + U_r \ot
  U_r^*}{2}$.  Averaging over $r$ we obtain
$$M^\dag M := \frac{1}{4}\sum_{r=1}^4 M_r^\dag M_r
= \frac{I + \frac{1}{8} \sum_{i=1}^8 U_i \ot U_i^*}{2}.$$
From the expansion properties proved in
\cite{QuantumMargulisExpanders}, we have that $\|M^\dag M  -
\proj{\phi_D}\| \leq \frac 1 2(1 + \frac{2\sqrt 5}{8}) = \frac{8 +
\sqrt 5}{16}$.
\end{proof}

To help prepare for the next sections, it is useful to view 
this test also in matrix form, as follows. 
If the initial state of the left/right registers was $\ket{x}_L\ket{y}_R$, after Alice's
operation, the state has to have the form
\begin{align}
	\frac{1}{\sqrt{d}}\sum_{i=1}^d \ket{i} (U_i \ket{x}_L) \ket{y}_R.
\end{align}
After Bob gets the ancilla and performs his operation $W^*$, the state has to have the form
\begin{align}
	\frac{1}{\sqrt{d}}\sum_{i=1}^d \ket{i} \left(U_i \ket{x}_L\right) \left(U_i^* \ket{y}_R\right).
\end{align}
Let us represent the initial state $\ket{\psi}_{AB}=\sum_{k,\ell}x_{k,\ell}\ket{k,\ell}$ by a matrix $X$, such that $X_{k,\ell}= x_{k,\ell}$.
We now rewrite the final state as a matrix with components $\beta_{(a,L), R}$:
\begin{align}
	\beta = \frac{1}{\sqrt{d}} \left[\begin{array}{c}
			U_1 X U_1^\dagger\\
			U_2 X U_2^\dagger\\		
			\vdots
 	\end{array}\right].
\end{align}
Passing the final test now means that 
\begin{align}
	U_i X U_i^\dagger = U_j X U_j^\dagger, \quad \forall i,j,
\end{align}
which is possible (if we have a quantum expander) only for $X=\ii$. 
This means the initial state was $\ket{\psi}_{LR} = \ket{\phi_D}$,
and that the final state is
$\left(\frac{1}{\sqrt{d}}\sum_{i=1}^{d} \ket{i}_a \right) \otimes \ket{\phi_D}$.

\section{A counterexample to the generalized area law}
\label{sec:HLMR}

In this Section we present our second result: a Hamiltonian with a small “bridge” term connecting two large halves of a system. Strikingly, this single-link bridge of constant dimensions has a large influence on the entanglement entropy between the two parts of the system, in the ground state.

\begin{figure}[h]
\begin{center}
\begin{tikzpicture}[scale=1.4] 
\tikzstyle{LabelStyle}=[fill=white,sloped]
\Vertex[x=0,y=0, style={minimum size=1.2cm},  LabelOut=false,L=$\Sigma_{L}$]{1} 
\Vertex[x=1.2,y=0,  LabelOut=false,L=$\sigma_{1}$]{2} 
\Vertex[x=2.4,y=0,  LabelOut=false,L=$\sigma_{2}$]{3}
\Vertex[x=3.6,y=0, style={minimum size=1.2cm}, LabelOut=false,L=$\Sigma_{R}$]{4}
\Edges[label={L}](1,2)
\Edges[label={M}](2,3)
\Edges[label={R}](3,4)
\end{tikzpicture}
\end{center}
\caption{\small A single-link chain with side operators $L$ and $R$.}
\label{fig-label}
\end{figure}
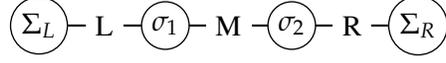

\medskip

\subsection{Background about local Hamiltonians} 
We consider Hamiltonians on
finite-size spin systems, where each term in the Hamiltonian is a
bounded-strength interaction between a bounded number of spins; in
fact our constructions will involve only pairwise interactions.
Unlike many physical systems, we do not require spatial locality but
allow interactions between any pair of spins.  We will also consider
systems in which individual spin dimension $d_i$ can be large.

A quantum state on $n$ particles, of dimensions $d_1,..,d_n$ respectively,
is a unit vector in $\cH := \bbC^{d_1} \ot \dots \ot \bbC^{d_n}$.  A
Hamiltonian $H$ is a Hermitian matrix acting on $\cH$.  A $k$-local
Hamiltonian can be written as $H = \sum_{i=1}^m H_i$ where each $H_i$
acts nontrivially on at most $k$ particles.
Conventionally,
each $\|H_i\| \leq 1$.  If the $H_i$'s are all diagonal, $H$ is
equivalent to a classical constraint satisfaction problem; in
general the $H_i$ may not always be diagonal or even commute.  The
eigenvector of $H$ with the smallest eigenvalue is called the ground state. 
We say
that a Hamiltonian is {\em frustration free} if the ground state of $H$ is
also the eigenvector with lowest eigenvalue for each $H_i$; otherwise
we call $H$ frustrated.

If the eigenvalues of $\cH$ are $E_0 \leq E_1
\leq \cdots \leq  E_{D-1}$ 
then the {\em gap} of
$H$ is defined to be $E_1 - E_0$.  When $H$ belongs to a family of
Hamiltonians indexed by $n$, we say this family is {\em gapped} if the gap
of each $H$ is lower-bounded by a constant independent of $n$.
(Otherwise the family is said to be {\em gapless}.)  Often we identify
$H$ with the family of Hamiltonians, and simply say that $H$ itself is gapped
or gapless.

\subsection{Construction of the Hamiltonian} 
\label{sec:Hconstruct}

Let the system $W$ consist of two qutrits ($\sigma_{1}$ and $\sigma_{2}$) and two high dimensional systems
($\Sigma_{L}$ and $\Sigma_{R}$):
\[
W = \Sigma_{L}\otimes \sigma_{1} \otimes \sigma_{2} \otimes \Sigma_{R} = C^{D}\otimes C^{3} \otimes C^{3} \otimes C^{D}. 
\]

We design a gapped Hamiltonian $H = H_L+ H_M+ H_R$, where  $H_L$ (left) $H_M$ (middle) and $H_R$ (right) 
are projectors acting on $\Sigma_{L}\otimes \sigma_{1} $, $\sigma_{1} \otimes \sigma_{2}$ and $\sigma_{2} \otimes \Sigma_{R}$, 
respectively, that defies the area law
through the cut $\Sigma_{L}\otimes \sigma_{1} \, \mid\, \sigma_{2} \otimes \Sigma_{R}$. For convenience we write all elements
of $W$ in the form
\[
\sum_{i,j\in[3]} |\psi_{i,j}\rangle|i\rangle |j\rangle,
\]
where $\psi_{i,j} \in \Sigma_{L} \otimes \Sigma_{R}$, and $\sum_{i,j\in[3]} |\psi_{i,j}|^2 = 1$. If we fix a basis in $\Sigma_{L}$ and $\Sigma_{R}$, respectively,
we can think of $\psi_{i,j}$ for every $i,j\in [3]$ as $D\times D$ matrices. Our construction will rely on
quantum expanders using $D\times D$ unitary matrices with $U_1=I$, and $U_{2}$ and
$U_{3}$ such that for any $D\times D$ matrix $X$ with  
$|X|^2 = \tr(XX^\dag) = 1$, $\tr(X)=0$ we have:
\begin{eqnarray}\label{eq:qexp}
|\cE(X)| = 
\frac 1 3 |X + U_{2}XU_{2}^\dag +  U_{3}XU_{3}^\dag| \le (1-c),
\end{eqnarray}
where $c:=1-\lambda>0$ is a fixed constant, independent of $D$. 
Equation (\ref{eq:qexp}) and the triangle inequality imply that
for any $D\times D$ matrix $X$ with 
$|X|= 1$,  $\tr(X)=0$:
\begin{eqnarray}\label{eq:qexp2}
|U_{2}XU_{2}^\dag -X| +  |U_{3}XU_{3}^\dag -X| \ge 3c
\end{eqnarray}

We now define projectors $H_L$, $H_R$ and $H_M$ via their zero subspaces 
${\cal L}$, ${\cal R}$ and ${\cal M}$. We describe these subspaces by writing states of $W$ in the block matrix form
\[
\left(\begin{array}{lll}
\psi_{1,1} & \psi_{1,2} & \psi_{1,3} \\
\psi_{2,1} & \psi_{2,2} & \psi_{2,3} \\
\psi_{3,1} & \psi_{3,2} & \psi_{3,3} \\
\end{array}\right).
\]
Note that our way to present a (pure) state of $W$ is unlike the density matrix presentation, and it is only meaningful, because $W$ is a tensor product
of four components. The above matrix form (of a vector) is 
simply a convenient way of rendering the $(3 D)^2$ coordinates of a state in $W$. In this presentation ${\cal L}$, ${\cal R}$ and ${\cal M}$
have convenient expressions.
${\cal L}$ is the set of states of the form
\[
\left(\begin{array}{rrr}
\psi_{1,1} & \psi_{1,2} & \psi_{1,3} \\
U_{2}\psi_{1,1} & U_{2}\psi_{1,2} & U_{2}\psi_{1,3} \\
U_{3}\psi_{1,1} & U_{3}\psi_{1,2} & U_{3}\psi_{1,3} \\
\end{array}\right),
\]
where $\psi_{1,1}$, $\psi_{1,2}$ and  $\psi_{1,3}$ are arbitrary.
${\cal R}$ is the set of states of the form
\[
\left(\begin{array}{rrr}
\psi_{1,1} & \psi_{1,1}U_{2} & \psi_{1,1}U_{3} \\
\psi_{2,1} & \psi_{2,1}U_{2} & \psi_{2,1}U_{3} \\
\psi_{3,1} & \psi_{3,1}U_{2} & \psi_{3,1}U_{3} \\
\end{array}\right),
\]
where $\psi_{1,1}$, $\psi_{2,1}$ and $\psi_{3,1}$ are arbitrary.
${\cal M}$ is the set of states of the form
\[
\left(\begin{array}{lll}
\psi_{1,1} & X & Y \\
X & \psi_{2,2} & \psi_{2,3} \\
Y & \psi_{3,2} & \psi_{3,3} \\
\end{array}\right),  \hspace{0.5in}  
\]
where $X$, $Y$ and the remaining $\psi_{i,j}$'s are arbitrary.
It is easy to check that $H_L, H_M, H_R$ are indeed local.  
For instance 
${\cal M}$ is a tensor product of $\Sigma_{L} \otimes \Sigma_{R}$ with the 
subspace $S$ of $\sigma_{1}\otimes \sigma_{2}$ that equates the coefficients of
$|1\rangle |2\rangle$ and $|2\rangle |1\rangle$ and also the
coefficients of $|1\rangle |3\rangle$ and $|3\rangle |1\rangle$.
Explicitly
\begsub{H-terms}
H_L &:= \bbI - \frac 1 3 \sum_{i,i' = 1}^3 U_i U_{i'}^\dag \ot \ket i \bra i', \\
H_R &:= \bbI - \frac 1 3 \sum_{j,j' = 1}^3 
	 \ket j \bra j' \ot \left(U_j U_{j'}^\dag\right)^T,  \\
H_M &:= \frac{(\ket{12}-\ket{21})(\bra{12}-\bra{21})}{2} + 
	\frac{(\ket{13}-\ket{31})(\bra{13}-\bra{31})}{2}.\label{HM}
\endsub

\subsection{The ground state is highly entangled} 
\label{sec:entangledG}

\begin{lemma}\label{lem:ground1}
The unique
normalized ground state $\ket{G}$ of $H = H_L + H_R + H_M = (\bbI -
\Pi_{\cal L}) +  (\bbI - \Pi_{\cal R}) +  (\bbI - \Pi_{\cal M})$ written out
as a matrix is 
\[
G = \frac{1}{3\sqrt{D}}  \left(\begin{array}{lll}
I & U_{2} & U_{3} \\
U_{2} & U_{2}^{2} & U_{2}U_{3} \\
U_{3} & U_{3}U_{2} & U_{3}^{2} \\
\end{array}\right).
\]
It satisfies $H\ket{G} = H_L\ket{G} = H_M\ket{G} = H_R\ket{G} =0$.
\end{lemma}

\begin{proof}
Equation (\ref{eq:qexp2}) guarantees that $I$ is the only $D\times D$ matrix that commutes with 
both $U_{2}$ and $U_{3}$. From this together with the above forms of 
 ${\cal L}, {\cal R}$ and  ${\cal M}$, 
 it follows that $\ket{G}$
is the only normalized state vector in ${\cal L} \cap {\cal R} \cap {\cal M}$.
\end{proof}

\begin{lemma}
The entanglement entropy of $\ket{G}$
along the $\Sigma_{L}\otimes \sigma_{1} \, \mid\, \sigma_{2} \otimes \Sigma_{R}$ cut
is
 $\log_{2} D$.
\label{lem:entangledstate}
\end{lemma}
\begin{proof}
Let $\ket{Z}$ be any state on our-four particle system, and let $Z$ be its matrix notation. 
By a direct calculation, the reduced density matrix of $\ket{Z}$ on the $\Sigma_L \ot \sigma_1$ systems is exactly   
$ZZ^\dag$. Since $G = \frac{1}{3\sqrt{D}} \sum_{i,j=1}^3 U_iU_j \ot
\ket i \bra{j}$ (letting $U_1 := \bbI$), we have that for the ground state $\ket{G}$ the reduced density matrix 
\be 
	GG^\dag = \frac{1}{3D} \sum_{i,i'=1}^3 U_iU_{i'}^\dag \ot \ket i\bra{i'} .
\ee
To diagonalize $GG^\dag$, let $W := \sum_{i=1}^3 U_i \ot \proj i$.
Then
\be 
	W^\dag (GG^\dag) W = \frac{\ii}{D} \ot \frac 1 3 \sum_{i,i'=1}^3 \ket i \bra{i'},
\ee
which has $D$ eigenvalues equal to $1/D$.
\end{proof}

\subsection{The Hamiltonian is gapped} 
\label{sec:Hgapped}

\begin{lemma}
\label{lem:gap1}
Denote the energy gap above the ground space for the Hamiltonian
$H=H_L+H_M+H_R$ by $\Delta$. Then $\Delta\geq c/4$ with $c$ defined in \eq{qexp}.
\end{lemma}

First we state a Lemma about the spectrum of the sum of two projectors.
\begin{lemma}\label{lem:two-proj}
Let $P_1,P_2$ be projectors onto subspaces $V_1, V_2$.  Let 
\be \mu := 
\min\{\bra{\psi}P_2\ket{\psi}  : \ket{\psi}\in V_1,
\braket{\psi}{\psi}=1\}
= \lambda_{\min}(P_1 P_2 P_1).\ee
Then the minimum eigenvalue of $\bbI-P_1 + P_2$ is $1-\sqrt{1-\mu}\geq
\mu/2$.
\end{lemma}
\begin{proof}
By Jordan's Lemma~\cite{Jordan75}, it suffices to consider the case
when
\be P_1 = \bpm 1 & 0 \\ 0 & 0 \epm
\qquad\text{and}\qquad
P_2 = \bpm \mu & \sqrt{\mu(1-\mu)} \\
\sqrt{\mu(1-\mu)} & 1-\mu\epm\ee
In this case $P_1 + P_2 = I - (1-\mu)\sigma_z + \sqrt{\mu(1-\mu)}\sigma_x$
which has eigenvalues $1 \pm \sqrt{1-\mu}$.
\end{proof}

\begin{proof}[Proof of \lemref{gap1}]
Let $V_{LR}$ be the ground space of $H_L + H_R$ and let $\tilde V_{LR}$ be the
subspace of $V_{LR}$ that is orthogonal to $\ket{G}$.  Let
$P_{LR}, \tilde P_{LR}$ be the corresponding projectors and observe
that
\be H_L + H_R \geq \bbI - P_{LR}
\qand
H_L + H_R + \proj G \geq \bbI - \tilde P_{LR}.
\ee
Let $\ket\psi \in
\tilde V_{LR}$ be a unit vector.  Since $\ket\psi\in V_{LR}$, we can write
it in matrix form as 
\begin{align}
	\ket{\psi} = \left(\begin{array}{rrr}
	X & XU_{2} & XU_{3} \\
	U_{2}X & U_{2}XU_{2} & U_{2}XU_{3} \\
	U_{3}X & U_{3}XU_{2} & U_{3}XU_{3} \\
	\end{array}\right).  \label{eq:AXBform}
\end{align}
Since $\braket{G}{\psi}=0$ we additionally have $\tr[X]=0$.
From normalization we have $|X| = 1/3$.  Now we calculate
\be
\mu := \bra{\psi}H_M \ket{\psi} = 
\frac{|XU_2 - U_2X|  +|XU_3 - U_3X|}{2}
\stackrel{\eq{qexp2}}{\geq} \frac{|X|}{2} 3c
= \frac{c}{2}
\label{eq:HM-exp-high}\ee
Setting $P_1=\tilde P_{LR}$ and $P_2 = H_M$ we can now apply
\lemref{two-proj} and find that the minimum eigenvalue of 
$(\bbI - \tilde P_{LR}) + H_M$ is $\geq c/4$.
Finally, the second-smallest eigenvalue of $H_L+H_M+H_R$ is the
minimum of $\bra{\psi}H_L+H_M+H_R\ket{\psi}$ over all unit vectors
$\ket\psi$ satisfying $\braket{\psi}{G}=0$.  For such a vector we have
\be 
\bra{\psi}H_L+H_M+H_R\ket{\psi}
= \bra{\psi}(H_L+H_M+H_R+\proj G)\ket{\psi}
\geq \bra{\psi}\bbI - \tilde P_{LR} + H_M\ket\psi
\geq \frac c 4 \ee
Combined with \lemref{ground1} this shows that the gap is $\geq c/4$.
\end{proof}

\section{The abstract Hamiltonian can be implemented locally}
\label{sec:concreteH}

The Hamiltonian construction in Section~\ref{sec:HLMR} has very
interesting properties (a unique, very entangled ground state and a
constant gap), but the $H_L$ and $H_R$ terms act on particles of
arbitrary dimension.  Alternatively, we can think of them as being
nonlocal Hamiltonians for a system of qubits.  We now wish to
decompose them into local terms, acting on particles of dimension $O(1)$, 
while retaining their desirable properties. This is done in two steps: 
we first construct a Hamiltonian $H'_{LMR}$ 
with the desired 
properties except the interactions are of polynomial strength, and then we correct this unphysical assumption and derive our final 
Hamiltonian  $H^{\text{gadget}}_{LMR}$. 

We start by showing in Subsection \ref{sec:compute3} that $H_L$
can be made local.  We do this using Kitaev's
circuit-to-Hamiltonian construction applied to the circuit computing
the application of the expander, padded with polynomially many
identity gates at the end of the computation. This derives a local
Hamiltonian $H^{\text{Kit}}_{L}$ with ground states very close to the
ground states of $H_L$ in Section~\ref{sec:HLMR}, tensored with some
state in an additional ancilla register.  The price we pay in this construction
is an inverse polynomial gap instead of a constant one since Kitaev's
construction has an inverse polynomial gap. To avoid this, we multiply
the local interaction terms in $H_L^{\text{Kit}}$ by a polynomial
prefactor and arrive at a Hamiltonian $H'_{L}$ with a constant gap, as
stated in Claim~\ref{HclockClaim}. However, its terms have
polynomially large, unphysical norms.

Next, in Subsection~\ref{sec:double}, we show in
\thmref{HlocalProperties} that by using the above local construction
on 
both sides of the $4$-particle chain of Section \ref{sec:HLMR},
without changing the middle interactions, we arrive at a Hamiltonian
whose strength of the middle interaction remains $O(1)$, while its unique
ground state retains all the desired properties of of the
four-particle Hamiltonian $H$ from Section~\ref{sec:HLMR}. Note that
the interaction terms which are not in the middle are still of
polynomial strength.  This gives us a local Hamiltonian $H'_{LMR} =
H'_L + H'_R + H_M$ with a constant gap, and a unique, entangled ground
state, just as we had for $H$. We now wish to make the strength of the
interactions on both sides bounded as well.

In Theorem \ref{FinalClaim}, 
we decompose each high-norm local interaction term in $H'_L$ and $H'_R$ into many local, constant-norm terms, using the strengthening gadgets 
of \cite{NagajCao}. Thus, we end up with a local Hamiltonian $H^{\text{gadget}}_{LMR}$ with all the desired properties of $H_{LMR}$ from Section~\ref{sec:HLMR}. We note that once again a price is to be paid:
in our final local, bounded-interactions Hamiltonian, each particle is involved in polynomially many 2-body interactions. 
It remains open to make the degree of interaction bounded.

\subsection{Evaluating a quantum expander locally (3 computations in parallel)}
\label{sec:compute3}

Let us translate the Hamiltonian from Section~\ref{sec:HLMR} into a local one.
We start by mimicking $H_L$ by a sum of local terms.
The Hamiltonian $H_L$ acts on a space of dimension $3D$, and its ground states have form 
\begin{align}
	\ket{\Phi_x} = \frac{1}{\sqrt{3}} \sum_{i=1}^{3}
	\ket{i} \otimes U_i \ket{x}. \label{phiX}
\end{align}
We will now enlarge our system and
find a local Hamiltonian $H'_L$, whose ground state 
will be close to
\begin{align}
		\ket{\Phi_x}\otimes \ket{w}, \label{extraTensor}
\end{align}
with $\ket{w}$ some state of an extra register.

\begin{claim}\label{cl:localexpander}
  There exists a frustration-free local Hamiltonian with a constant spectral gap, set
  on a chain of $2N+2$ constant-dimensional qudits, such that all ground
  states are $\epsilon$-close to the form \eqref{extraTensor}, with
  $\epsilon$ inverse polynomial in $N$. The local terms of the
  Hamiltonian are of norm bounded by $O(\poly(N))$.
\label{HclockClaim}
\end{claim}

We prove Claim \ref{cl:localexpander} with $N'=\poly(n)$ qubits and
$5$-local interactions (in general geometry). This construction can
then be recast on a chain of $2N+2=\poly(n)$ qudits using \cite{AGIK}
or \cite{QMAcomplete1Dd8}.  
There, the clock/data registers (with $N$ particles each) can be seen as sitting on top of each other, and pair clock/data particles into larger qudits. These will then sit on a chain $b_N$-$\cdots$-$b_1$-$\sigma_1$-$\sigma_2$-$b'_1$-$\cdots$-$b'_N$,
with qutrits $\sigma_1$, $\sigma_2$ in the middle.

We construct the Hamiltonian of Claim~\ref{HclockClaim} following Kitaev's Circuit-to-Hamiltonian construction \cite{KitaevBook}. 
It allows us to write down a Hamiltonian whose ground states are the history states of a quantum computation $V$, i.e.
states of the form
\begin{align}
	\ket{\psi^{\text{hist}}_{y}} = \frac{1}{\sqrt{T}} \sum_{t=0}^{T} \ket{t}_k \otimes V_{t} \dots V_1 \left( \ket{y}\otimes \ket{0}_q \right),
\end{align}
where $k$ is an extra ``clock'' register, $q$ is an ancilla register, 
$\ket{y}$ is some initial state of a data register and $V_t$ are the gates of some circuit $V$, acting on the data register.
Our data register will contain $n$ data qubits (for simplicity, set $2^n = D$) and a ``control'' qutrit $a$. 
We want to get the history state of the circuit $V$ with unitaries 
\begin{align}
	V_t = \sum_{i=1}^{3} \ket{i}\bra{i}_a \otimes U_{i,t},
\end{align}
for $t=1,\dots,\tau$.  Here $U_{i,1},\ldots,U_{i,\tau}$ are the gates
that together implement $U_i$ from the quantum expander, including uncomputing
any changes to the ancilla register $q$ at the end.  On top of this, we pad the circuit $V$ with many identity gates $V_t = \ii$ for $\tau < t \leq T$, for some $T \gg \tau$, setting
$\epsilon = \frac{\tau}{T} = \frac{1}{\poly(n)}$.
We also require an extra clock register $k$ capable of locally implementing a clock with $T+1$ clock states, as well as an ancilla scratch register $q$. 
The ground states (history states of $V$) for the new Hamiltonian $H'_L$ 
have form
\begin{align}
		\ket{\Psi_x} &= 
		\frac{1}{\sqrt{T+1}} \sum_{t=0}^{T}
		\ket{t}_{k}	 \otimes \label{groundL}
V_t \dots V_1 \left(\frac{1}{\sqrt{3}} \sum_{i=1}^{3}  \ket{i}_a \ket{x}_{d} \ket{0\cdots 0}_{q} \right).
\end{align}
We will build $H^{\text{Kit}}_L = H_{\text{init}} + H_{\text{prop}}$ from two parts. 
First, propagation-checking:
\begin{align}
		H_{\text{prop}} &= \frac{1}{2} \sum_{t=1}^{T}
		\left(
				\ket{t-1}\bra{t-1}_k + \ket{t}\bra{t}_k
				\right) \label{Hproj}
- \frac{1}{2} \sum_{t=1}^{T}
		\left(
				\ket{t-1}\bra{t}_k \otimes V_t^\dagger 
				+ \ket{t}\bra{t-1}_k \otimes V_t
				\right).				
\end{align}
Second, we need to ensure proper initialization by adding a projector that prefers a uniform superposition on the control qutrit when the clock register is $\ket{0}_k$ (we want all three computations to run on the same input). 
Adding standard ancilla initialization-checking, we get
\begin{align*}
		H_{\text{init}} &= \ket{0}\bra{0}_{k} \otimes \left[	
		\ii - \ket{\alpha_3}\bra{\alpha_3}_{a} 
		+ \sum_{i=1}^{s} \ket{1}\bra{1}_{q_i} \right],
\end{align*}
with $\ket{\alpha_3} = \frac{1}{\sqrt{3}}\left(\ket{1}+\ket{2}+\ket{3}\right)$. 
We can now implement the clock register and the corresponding projectors by a a 5-local,
unary clock with $T+1$ qubits \cite{KitaevBook}.
The Hamiltonian $H^{\text{Kit}}_L$ is positive-semidefinite, and frustration-free. It has a zero-energy state of the form \eqref{groundL} for any basis state $\ket{x}$ of the $n$ working qubits. Furthermore, the energy gap of $H^{\text{Kit}}_L$ to eigenstates with nonzero energy is $\Delta^{\text{Kit}}_L = \oo{T^{-2}}$ \cite{KitaevBook}.
Using the 1-D construction for a line of constant ($8$-dimensional) qudits from \cite{QMAcomplete1Dd8} based on \cite{AGIK}, which also has a gap that scales as an inverse polynomial in $T$, this results in a 1-D Hamiltonian with the properties we want.

Let us consider the ground states more closely.  
For $t\geq \tau$, 
the data register is in the desired state $\ket{\Phi_x}$ \eqref{phiX},
the ancilla register is uncomputed, and it is only the 
clock register that changes. Recalling $T\gg \tau$, we realize that each $\ket{\Psi_x}$ can be rewritten as
\begin{align}
	\ket{\Psi_x} &= \frac{1}{\sqrt{T}}
	\sum_{t=1}^\tau \ket{\varphi_{x,t}}
+  \ket{\Phi_x}_{cd}\otimes  \frac{1}{\sqrt{T}} \left(\sum_{t=\tau+1}^T \ket{t}_k \right)				
				\otimes \ket{0\cdots 0}_q
 \nonumber\\
 &= \sqrt{\epsilon} \, \ket{v_{x}} + 
					\sqrt{1-\epsilon} \, \ket{\Phi_{x}}_{cd} \otimes \ket{w},  \label{psiL}
\end{align}
with some normalized
vectors $\ket{v_x}$ and $\ket{w}$. 
Each ground state is thus as close to $\ket{\Phi_x}\ket{w}$ \eqref{extraTensor} as we want, because we are free to choose $T$ as large a polynomial 
as we want, making 
$\epsilon=\tau/T$ an inverse polynomial as small as we want. 

The gap of the Hamiltonian $H^{\text{Kit}}_L$ is however not constant. 
We rescale the interaction strengths of all terms in $H^{\text{Kit}}_L$ by $T^2$ (or by a higher polynomial in $T$ for the 1-D construction), 
and look at $H'_L = \poly(T) \cdot H^{\text{Kit}}_L$. This new $H'_L$ satisfies the requirements of 
Claim~\ref{HclockClaim}.

\subsection{A local Hamiltonian with an entangled ground state}
\label{sec:double}

We now take two copies of the system from the previous Subsection, and
construct a Hamiltonian $H'_{LMR} = H'_L + H'_R + H_M$.   We keep the
same two-qutrit middle term $H_M$ from Eq.~(\ref{HM}) in
\secref{HLMR}, but will replace the left and right terms $H_L,H_R$
with the construction from Subsection~\ref{sec:compute3}.
\begin{theorem}
\label{thm:HlocalProperties}
The 1-D qudit Hamiltonian $H'_{LMR} = H'_L + H'_R + H_M$ with terms of
norm $\poly(n)$ has a unique ground state,  whose entanglement accross the middle cut
is at least $\Omega(\log(D))$, and a constant energy gap.
\end{theorem}

Unlike in \secref{HLMR}, this Hamiltonian is no longer frustration
free.  However, a qualitatively similar version of the argument from
that section will work.  One change is that we will work with an
approximate ground state.   Define
\begin{align}
	\ket{G'} &:= \frac{1}{\sqrt{D}}\sum_{x=1}^D \ket{\Phi_x}\ket{\Phi_x}\ket{w}\ket{w}
	\label{groundLMR}\\
	&= 
	\frac{1}{3\sqrt{D}} \left[\begin{array}{rrr}
		\ii & U_2 & U_3 \\
		U_2 & U_2^2 & U_2 U_3 \\
		U_3 & U_3 U_2 & U_3^2
 	 	\end{array}\right] \otimes \ket{w}\otimes\ket{w} =\ket{G}\ket{w}\ket{w},
	\nonumber
\end{align}
with $\ket{\Phi_x}$ from \eqref{phiX}, $\ket{G}$ from
Lemma~\ref{lem:ground1}, and $\ket{w}$ a state of an ancilla register.
Thus the state $\ket{G'}$ is exactly the ground state we have in
Section~\ref{sec:entangledG}, with ancilla states added.  It is not
the ground state of $H'_{LMR}$, nor can we even prove that it has low
energy. However, we will later construct a state $\ket{G_\eps'}$ that
both has low enough energy to be close to the true ground state, and
is close enough to $\ket{G'}$ to have large entanglement.

The rest of our argument breaks up into the following subsidiary claims.
\begin{claim}\label{clm:GHG-small}
$\bra{G_\eps'}H'_{LMR}\ket{G_\eps'} \leq 1/\poly(n)$ (with
$\ket{G'_\eps}$ defined later).
\end{claim}
\begin{claim}\label{clm:gap-large}
 The second-smallest eigenvalue of $H'_{LMR}$ is $\geq
  \Omega(1)$, implying that the gap is large.
\end{claim}
\begin{claim}\label{clm:entangled}
The ground state of $H'_{LMR}$ has large overlap with
  $\ket{G'}$, and therefore high entanglement.
\end{claim}

We begin by showing a precise sense in which $H_L', H_R'$ give an
approximation of $H_L,H_R$.
Define $H'_{LR}:=H'_L+H'_R$
to be the Hamiltonian acting on the two
sides of the chain without interaction.  The ground space
of $H'_{LR}$ is spanned by basis states of the form
\begin{align}
	\ket{\Psi_x}\ket{\Psi_y} &=
			\sqrt{\epsilon(2-\epsilon)}\,  \ket{z^{\epsilon}_{xy}}
			+ (1-\epsilon) \ket{\Phi_x}\ket{\Phi_y}\ket{w}^{\otimes 2}, 		
			\label{groundLReps}
\end{align}
where $\ket{\Psi_x}$ and likewise $\ket{\Psi_y}$ are given by
\eqref{psiL}, and $\epsilon = \frac{\tau}{T}$ is an inverse polynomial
which we can make as small as we want by increasing $T$ to a large
polynomial in $n$.

\begin{definition}\label{defs0} {\bf $S_0, S_\epsilon$:}
  Define $S_0$ to be the space spanned by states of the form
  $\ket{\Phi_x}\ot\ket{\Phi_y}$, where $\ket{\Phi_x}$ (and likewise
  $\ket{\Phi_y}$) are defined in \eqref{phiX}.  Define $S_\epsilon$ to
  be the space spanned by all states of the form with $\ket{\Psi_x}
  \ot \ket{\Psi_y}$, defined in \eqref{psiL}.  Define the
  corresponding projectors to be $P_0, P_\eps$.
\end{definition} 

\begin{claim} \label{clm:s0} 
Let $\ket{s_\epsilon}$ be any state in
  $S_\epsilon$; then there exists a state $\ket{s_0}\in S_0$ such that
\begin{equation} 
	\ket{s_\epsilon}=\sqrt{1-\epsilon'^2}\ket{s_0}\ket{w}\ket{w}+\epsilon'\ket{\epsilon},
\end{equation}
for $\epsilon'\le 2\epsilon$, and $\ket{\epsilon}$ orthogonal to
$\ket{s_0}\ket{w}\ket{w}$.  As a result \be \frac{1}{2}\norm{
  \proj{s_0} - \proj{s_\eps}}_1 \leq 2\eps .
\label{eq:s0-seps}\ee
\end{claim} 
\begin{proof}
  The proof follows by direct calculation, using Definition
  \ref{defs0} of $S_0$ and $S_\epsilon$ and the observation that the
  normalized error vectors $\ket{z_{xy}^\epsilon}$ from \eqref{psiL}
  are all orthogonal for different pairs of $x,y$. This follows from
  their definition as history states of different initial vectors, as
  seen in \eqref{psiL}.
\end{proof}

As a direct consequence of \clmref{s0} we establish
\clmref{GHG-small}.  Indeed, $\ket{G'}\in S_0 \ot \ket w \ot \ket w$
and by \clmref{s0} there exists an $\eps$-close state $\ket{G_\eps'}\in S_\eps$ in
the ground space of $H_{LR}'$.  Thus
\be \bra{G'_\eps}H_L' + H_M + H_R'\ket{G'_\eps} = 
\bra{G'_\eps}H_M\ket{G'_\eps} \leq 2\eps \leq 1/\poly(n),
\label{eq:Geps-low-E}\ee
where in the last step we have used \eq{s0-seps} and $0\leq H_M\leq I$.

Now we consider gap.
The Hamiltonians $H'_L$ and $H'_R$ act on independent subspaces, while
both Hamiltonians have a constant energy gap above the ground state
subspace as we proved in Claim \ref{cl:localexpander}. 
Therefore, $H'_{LR}$ has constant gap. Let us denote this
gap by $\Delta'_{LR}$, so that (using also the fact that $H_{LR}'$ has
lowest eigenvalue 0) we have the operator inequality
\be H'_{LR} \geq \Delta'_{LR} (\bbI - P_\eps).
\label{eq:prime-op-ineq}\ee
Continuing along the lines of the proof of \lemref{gap1}, define
$\tilde P_\eps := P_\eps - \proj{G_\eps'}$, $\tilde P_0 = P_0 - \proj{G}$
and define $\tilde S_\eps, \tilde S_0$ to be their supports.  
Now calculate
\begsub{mu-prime}
\mu &:= \min \{\bra{\psi}H_M\ket{\psi} : \ket\psi \in \tilde S_\eps,
\braket{\psi}{\psi}=1\} \\
&\stackrel{\eq{s0-seps}}{\geq}
 \min \{\bra{\psi}H_M\ket{\psi} : \ket\psi \in \tilde S_0\ot\ket{w}\ot\ket{w},
\braket{\psi}{\psi}=1\} - 2\eps\\
&\stackrel{\eq{HM-exp-high}}{\geq} \frac{c}{2} - 2\eps
\endsub

Denote the second-smallest eigenvalue of a Hermitian matrix $X$ by
$\lambda_2(X)$.  A variant of the Courant-Fischer min-max principle gives
the following variational characterization of $\lambda_2(X)$:
\be \lambda_2(X) = \sup_{\ket{v}} \lambda_{\min}(X + \proj v).\ee
We can apply this to our problem by observing that
\ba 
\lambda_{2}(H_{LMR}') 
& = \sup_{\ket{v}} \lambda_{\min}(H_{LMR}' + \proj{v}) 
\\ & \geq \lambda_{\min}(H_{LMR}' + \Delta'_{LR} \proj{G_\eps'}) 
\\ & \stackrel{\eq{prime-op-ineq}}{\geq}
\Delta'_{LR} \lambda_{\min}(\bbI - \tilde P_\eps + H_M) 
\\ & \stackrel{\text{\lemref{two-proj}}}{\geq}\;
\frac{c}{4}-\eps \geq \Omega(1)
\ea
This establishes \clmref{gap-large}.

To complete the proof, we need to show that the ground state is highly
entangled.  Two challenges which complicate the usual continuity
arguments are that $H_{LMR}'$ has a large norm and a large ancilla
dimension.  We sidestep these as follows.
Let $\ket{\gamma}$ denote the ground state of $H_{LMR}'$. Adjust its
overall phase so that $\braket{\gamma}{G_\eps'}$ is real, implying
$\ket{G_\eps'} = \sqrt{1-\delta}\ket\gamma +
\sqrt{\delta}\ket{\gamma^\perp}$ for some orthogonal state
$\ket{\gamma^\perp}$.  Then
\be \eps \stackrel{\eq{Geps-low-E}}{\geq}
\bra{G_\eps'}H'_{LMR}\ket{G_\eps'} = 
(1-\delta)\bra{\gamma}H'_{LMR}\ket{\gamma} +
\delta\bra{\gamma^\perp}H'_{LMR}\ket{\gamma^\perp}
\geq \delta \L(\frac c 4 -\eps\R),\ee
where this last inequality follows from the fact that $H_{LMR}'$ is
positive semi-definite and has gap $\geq c/4-\eps$.  Thus $\delta \leq
5\eps/c$ (assuming $\eps \leq c/20$).

Combining this with previous facts we have \be \ket\gamma
\approx_\delta \ket{G_\eps'}\approx_{2\eps} \ket{G'} =\ket{G}\ot\ket w
\ot \ket w\label{eq:approx-chain}\ee and the latter state is highly
entangled according to \lemref{entangledstate}.  Assume WLOG that
$\braket{\gamma}{G'}$ is real and nonnegative.  Thus $\braket{g}{G'}
\geq 1-\eps'$ for some $\eps' = \poly(1/n)$.  However, the large
dimension of the ancilla states means we cannot directly use Fannes's
inequality.  Let the Schmidt decomposition of $\ket{\gamma}$ be
$$\ket\gamma = \sum_i \sqrt{\lambda_i} \ket{L_i} \ot\ket{R_i},$$
with $\lambda_1 \geq \lambda_2 \geq \ldots \geq 0$ and $\sum_i
\lambda_i=1$.  
Then
$$1-\eps' \leq \braket{\gamma}{G'}
\leq \frac{1}{\sqrt{D}} \sum_{i=1}^D \sqrt{\lambda_i}.$$
Let $r$ be the largest $r$ for which $\lambda_r \geq \kappa/D$ for
$\kappa>1$ to be chosen later.  By normalization, $r \leq D/\kappa
\leq D$.   Let $\beta = \sum_{i\leq r} \lambda_i$.
Then
$$\sqrt{D}(1-\eps') \leq \sum_{i=1}^D \sqrt{\lambda_i}
\leq \sqrt{\frac{D}{\kappa}}\sum_{i\leq r}\lambda_i + 
\sqrt{D}\sum_{i=r+1}^D \lambda_i
\leq \sqrt{D}(\beta\kappa^{-1/2} + (1-\beta)).$$
The second inequality follows from $\sqrt{\lambda_i}\leq
\sqrt{\frac{D}{\kappa}}\lambda_i$ for $i\leq r$ in the first term and
  Cauchy-Schwarz in the second term.
Rearranging we find that $\beta \leq \eps'(1-\kappa^{-1/2})$.
We conclude that the entropy of entanglement is
$$\sum_i \lambda_i \log\L(\frac{1}{\lambda_i}\R)
\geq \sum_{i>r}\lambda_i \log\L(\frac{D}{\kappa}\R)
\geq (1-\eps'(1-\kappa^{-1/2})) \log\L(\frac{D}{\kappa}\R).$$
Optimizing over $\kappa$ we find that the entanglement is
$(1-o(1))\log(D)$.
This concludes the proof of \clmref{entangled} and therefore \thmref{HlocalProperties}.

\subsection{Decomposing the Hamiltonian $H'_L$ into $O(1)$-strength interaction terms}
\label{sec:gadgets}

We now handle the problem of large interaction norm.  The interactions
in the Hamiltonian $H'_L$ have norm
$\poly(T)$.  Each such term can be decomposed using the 
  strengthening quantum gadget construction by Nagaj and
Cao~\cite{NagajCao}, into $\poly(T)$ 
bounded-strength interactions acting on the original set of qudits
plus $\poly(T)$ extra ancilla qubits.  The gap of this new $H'_L$ will
remain a constant, while any state in its ground state will now be
$1/\poly(T)$ close to some $\ket{\Psi_x} \ket{0\cdots 0}_{\textrm{new
    ancillas}}$, with $\ket{\Psi_x}$ from \eqref{psiL}. This also
implies that each (less than a small constant energy) state of $H'_L$
is $1/\poly(T)$ close to the state $\ket{\Phi_x} \ket{w} \ket{0\cdots
  0}_{\textrm{new ancillas}}$ for some $x$. However, this is just what
we had in \eqref{extraTensor}, with an expanded ancilla register state
$\ket{w'} = \ket{w}\ket{0 \cdots 0}$. Therefore, all of the arguments
of Section~\ref{sec:double} go through, and we have shown that
\begin{theorem}
There exists a 2-body Hamiltonian on $n$ qudits, whose terms are of
 $O(1)$ norm. The interaction graph is as in Figure~\ref{fig-intro},
where the two particles on the two sides of the cut 
are qutrits. All particles are involved in at most 
$\poly(n)$ interactions. Moreover,  
the Hamiltonian is gapped with a unique ground state, 
such that the entanglement entropy across the middle 
cut scales as $\Omega\left(n^{c}\right)$ for some $0<c<1$.
\label{FinalClaim}
\end{theorem}

\section{Entanglement testing and ground states of Hamiltonians}
\label{sec:comm}

\begin{figure}
\begin{center}
\includegraphics[width=13cm]{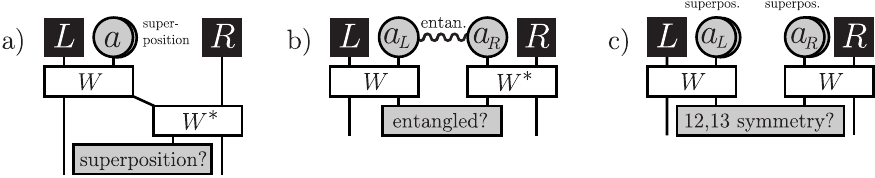}%
\caption{EPR testing procedures. a) The test from Section~\ref{sec:EPRprotocol}.
b) A modified test with two ancillas.
c) The Hamiltonian from Section~\ref{sec:HLMR} can be also easily recast as an EPR testing protocol.}%
\label{fig:commgames}\end{center}
\end{figure}

We now connect our two results more directly, by providing an alternative derivation of the results in Section \ref{sec:HLMR}. 
Starting from the 
EPR testing protocol 
 of Section~\ref{sec:EPRprotocol}, we turn it into a non-local Hamiltonian violating 
the generalized area law, using 
Kitaev's circuit-to-Hamiltonian construction. 
In fact, we use a slight variant
 of the EPR testing protocol, which uses two ancillas (see 
 Figure~\ref{fig:commgames}b), as it translates to a Hamiltonian more easily.

We first describe the modified EPR testing protocol. 
Alice has two registers, $L\otimes a_L$, and Bob has two registers denoted 
$a_R\otimes R$, where $R,L$ are of large dimension $D$ and 
$a_R$,$a_L$ are of constant dimension $d$. 
They wish to check whether their joint state $\ket{\psi}_{LR}$ on registers $L \otimes R$ is maximally entangled. 
First, Alice and Bob pre-share a maximally entangled state on $a_L\otimes a_R$: 
\begin{align}
	\ket{\phi_{d}} = \frac{1}{\sqrt{d}}\sum_{i=1}^{d}\ket{i}_{a_L}\ket{i}_{a_R} \label{entancillas}
\end{align}
Second, Alice applies the unitary $W=\sum_{i=1}^{d} \proj i \ot U_i$
to $a_L \otimes L$, and Bob applies $W^*$ to $a_R\otimes R$. 
Finally, they apply a projective measurement on $a_L\otimes a_R$ 
of the state $\ket{\phi_d}$ \eqref{entancillas}.
It is not difficult to see that this too is an EPR testing protocol; the
test passes with probability close to 1 if and only if the original state $\ket{\psi}_{LR}$ 
was very close to the maximally entangled state.

To encode this protocol into a Hamiltonian 
via the circuit-to-Hamiltonian construction, we use {\it two} independent, two-step clocks. 
(We will think of the circuit $W$ as well as $W^*$ as applied in a single
 time step). 
The Hamiltonian will thus act on four registers, $L,R$ and two enlarged 
registers, $A_L=a_L\otimes C_L$ and $A_R= a_R\otimes C_R$ with $C_L,C_R$ 
being the two $2$-dimensional spaces of the two clocks, respectively.  
We write the basis states of $A_L, A_R$ as $\ket{0,i}$ and $\ket{1,i}$ 
for $i\in \{1,..,d\}$. 

The Hamiltonian consists of the following terms. An ``initialization'' and ``output'' term 
on $A_L\otimes A_R$: 
\begin{align}
		H_{M} &= \sum_{s = 0}^{1}\sum_{i,j = 1}^{d} \ket{s,i}\bra{s,i}_{A_L}\otimes 
 \ket{s,j}\bra{s,j}_{A_R}
- \sum_{s = 0}^{1}\sum_{i = 1}^{d} \ket{s,i}\bra{s,i}_{A_L}\otimes \ket{s,i}\bra{s,i}_{A_R}
\end{align}
whose ground states have the form $\ket{1,i}\ket{0,j}$ and $\ket{0,i}\ket{1,j}$ for any $i,j$, but more importantly 
$\frac{1}{\sqrt{d}}\sum_{i=1}^{d} \ket{0,i}\ket{0,i}$ and $\frac{1}{\sqrt{d}}\sum_{i=1}^{d} \ket{1,i}\ket{1,i}$.
These two states are maximally entangled states of the ancillas when 
the ``clocks'' are both $0$ (initialization) or both $1$ (output).

Second, we have the  ``left-computation-checking'' Hamiltonian, which acts on the registers $a_L$ and $L$: 
\begin{align}
	H_L &= \frac{1}{2} \sum_{i=1}^{d} \left(\ket{0i}\bra{0i} +
          \ket{1i}\bra{1i}\right)_{a_L} \otimes \ii_{L}
- \frac{1}{2} \sum_{i=1}^{d} \ket{1i}\bra{0i} \otimes W
				- \frac{1}{2} \sum_{i=1}^{d} \otimes
                                \ket{0i}\bra{1i} \otimes W^\dagger.
\end{align}
Similarly, we define $H_R$, the ``right-computation-checking'' Hamiltonian which acts on $A_R$ and $R$, replacing $W$ by $W^*$ and $W^\dagger$ by $W^T$:
\begin{align}
	H_R &= \frac{1}{2} \sum_{i=1}^{d} \left(\ket{0i}\bra{0i} +
          \ket{1i}\bra{1i}\right)_{a_R} \otimes \ii_{R}
- \frac{1}{2} \sum_{i=1}^{d} \ket{1i}\bra{0i} \otimes W^*
				- \frac{1}{2} \sum_{i=1}^{d}  \ket{0i}\bra{1i} \otimes W^T.
\end{align}
The final Hamiltonian, $H=H_{M}+H_{L}+H_{R}$ is our desired counterexample. 
We claim that its unique, frustration-free ground state
is the ``history'' state
\begin{align}
	\ket{\Psi} &=
			\frac{1}{\sqrt{d}} \sum_{i=1}^{d} 
		\left(
						\ket{0,i}_{a_L} + \left(W\otimes \ii\right) \ket{1,i}_{a_L} 
		\right)
		\left(
						\ket{0,i}_{a_R} + \left(\ii \otimes W^* \right) \ket{1,i}_{a_R} 
		\right)
		\ket{\phi_{D}}_{LR}. \nonumber
\end{align}
It is not difficult to check that this is a maximally entangled state of 
dimension $dD$, by observing that the Schmidt rank is $dD$ 
and the coefficients are uniform.

\section{Discussion, related work and Open Questions} 
\label{sec:discussion}

We have shown that in both the context of EPR testing, as well as
ground states of Hamiltonians, constant resources suffice to enforce
what seems to be a global property.  Our results are reminiscent in
spirit to the classical PCP theorem, or more generally to property
testing.  The common theme is that a small amount of resources
(bits checked, Hamiltonian interactions, qubits transmitted, etc.)
serve to verify the properties of some large object. However, the fact
that such highly non-local properties as global entanglement can be
detected using local resources seems rather counter-intuitive, and
calls for further investigation in other contexts. 

Our results leave many questions open. Below we discuss them as well as the broader context of these results. 

{~}

\noindent{\bf The Area Law question}
Of course, the major open question of the 2-D area law, which was the
main motivation for this work, is left wide open.  A more modest goal
would be to reduce the degree in our construction to a constant. Such
a step already seems to require significant progress in our
understanding of related notions, e.g., parallel
circuit-to-Hamiltonian constructions (see e.g.,\cite{RecentTerhal}),
and quantum expanders which are geometrically constrained, as well as
the notion of quantum degree reduction, as a possible route towards
quantum PCP~\cite{AAV13}.  Alternatively, it might be true that the
generalized area law does in fact hold with bounded-degree
bounded-strength Hamiltonians.

Indeed, such a conjecture is not unplausible, and could potentially be
motivated by the following intuition.  The area law had been long
believed, without proof, to be related to another very important
physical property of gapped Hamiltonians: the exponential decay of
correlations in the ground state. This means that a Hamiltonian has an
associated correlation length $\xi$ such that
$|\ip{AB}-\ip{A}\ip{B}| \leq \|A\| \, \|B\| \, e^{-\ell/\xi}$ where
$\ip{X}:=\bra{G} X \ket{G}$ and $A,B$ are observables separated by a distance $\ell$. 
Such an exponential decay is 
known to hold on a lattice of any constant dimension, and in fact in
any constant-degree graph~\cite{HastingsK06}
\footnote{Why don't our constructions contradict this, since they will
  have large amounts of entanglement in the ground state?  Our
  large-dimension construction in \secref{HLMR} does fit the criteria
  of \cite{HastingsK06} to have constant correlation length, but there
  the entire graph has constant diameter.  Our construction in
  \secref{concreteH} has large degree and so the correlation-length
  bound from \cite{HastingsK06} is also growing with the system size.}
It is perhaps natural to conjecture that if correlations in gapped
Hamiltonians are in this way ``local'', entanglement is also local;
One way to quantify this is with the area law conjecture.  However, we
stress that only in 1-D this implication is known to hold~\cite{BrandaoH12}.

Our results (in particular \thmref{HlocalProperties}) provide a
counterexample to another possible version of the area law: one in
which the interaction degree is bounded, but the norms of the
interactions are required to be bounded only accross the cut, and
otherwise they can be polynomially large. The rationale of this
condition is that large norm terms on each side of the cut should only
increase the entanglement within the two regions on each side of the
cut and therefore by monogamy of entanglement only {\it decrease} the
entanglement across the cut. Our counterexample suggests that the
above monogamy-of-entanglement argument is too naive.

Another possibile version of the area law that might still hold is that a subsystem with
dimension $d_i$ at distance $\ell_i$ from the cut can contribute at most
$\log(d_i) e^{-\ell_i/\xi}$ entanglement, where $\xi$ is the correlation
length. 
Attempting to strengthen our counter-example may either rule
these conjectures out or, in failing to do so, give a hint of how
they might be proved. 

We note that the implications of an(y of the above forms of an) area
law for general systems are not yet fully understood.  For ground
states of gapped one-dimensional systems, proving an area law was an
important step towards proving that they can be efficiently
described~\cite{Hastings09} and that these descriptions can be found
efficiently~\cite{AradKLV12area,LandauVV14}.  For Hamiltonians on
general graphs, a partition into pieces with subvolume entanglement
scaling (i.e. region $S$ has $\leq \eps |S|$ entropy) would imply a
classical description accurate enough to be incompatible with the
quantum PCP conjecture~\cite{BH-product}.  Since entropy is a way to
count effective degrees of freedom, another interpretation of area
laws is that a quantum system can equivalently be represented by a
theory living on its boundary.  This idea is known as the holographic
principle, and is currently a major conjecture in quantum field
theory~\cite{Swingle12}. It remains to be clarified whether an area law of any of
the suggested forms can lead to a more succinct description of the
ground states.

{~}

\noindent{\bf Related work} 
There are several related works that we would like to mention.  First,
Gottesman and Hastings \cite{GottesmanH10}, Irani \cite{Irani10}, and
Movassagh and Shor \cite{MS14} have examined qudit chains with highly
entangled ground states for Hamiltonians whose gaps are
inverse-polynomial. To the best of our knowledge, our results cannot
be derived in a straightforward manner from these works. Here, we
focus on spin chains with a constant gap. One can attempt to get a
constant-gap version of the above constructions by using the
strengthening gadgets of Nagaj and Cao~\cite{NagajCao} as we did in
this paper. However, this fails to provide the desired counterexample,
since these gadgets introduce a complicated geometry of interactions,
and we would need to apply them for every edge. Thus, the size of the
cut in the resulting graph would no longer be small.  It is crucial
that in our present construction, the middle link is unchanged; only
the rest of the interactions need to be strengthened by gadgets.

We mention another relevant prior work \cite{Hastings-expander1, Hastings-notes}, which
described a state on a one-dimensional chain with a large amount of
entanglement across cuts (say $\log n$) but only short-range correlations.   The
claim about decaying correlations here is rather subtle: two regions
that are separated by a distance $\ell$ from each other and $\ell'$
from the boundary have correlation no greater than $e^{-\ell/\xi}(1 + n
e^{-\ell'/\xi})$.  In this way it avoids contradicting the relation
between decaying correlation and area law from \cite{BrandaoH12}.  See
\cite{Hastings-notes} for further discussion.  This result is
incomparable to ours because the states in question are not ground
states of a gapped $O(1)$-local Hamiltonian.

{~}

\noindent{\bf More general implications}
Finally, we believe that our results point at a potentially useful
link between two seemingly unrelated topics.  Our paper shows that a
counterexample to the generalized area law can be derived from an
entanglement testing protocol of limited communication and converting
it into a Hamiltonian using Kitaev's circuit-to-Hamiltonian
construction.  Our area-law violating Hamiltonian can be viewed as a
``tester'' of its highly entangled ground state, where the norm of the
Hamiltonian terms along the cut corresponds to the communication
complexity of the protocol.  Can any area-law-violating Hamiltonian be
connected to an entanglement-testing protocol with communication
pattern corresponding to the interaction graph of the Hamiltonian?
More generally, in what ways can Hamiltonians be viewed as {\it
  testers} for their ground states?  Whether such a ``translation''
always exists between entanglement testing protocols of limited
communication, and entangled ground states of Hamiltonians with
limited interactions between different parts of the system, remains to
be explored.  Making such an equivalence rigorous might open up a
whole new set of tools to studying the area law question, and more
generally, help develop better intuition for local Hamiltonians and
their ground states.  A related question is whether EPR testing is in
fact {\it equivalent} in some sense to the property of being a quantum
expander.

\section{Acknowledgements}
\label{sec:thanks}

The authors thank the Simons Institute (the Quantum Hamiltonian
Complexity program) where part of this work was done. DA acknowledges
the support of ERC grant 030-8301 and BSF grant 037-8574.  AWH was
funded by NSF grant CCF-1111382 and ARO contract W911NF-12-1-0486.  ZL
was supported by NSF Grants CCF-0905626 and CCF-1410022 and Templeton
Grants 21674 and 52536. DN thanks the Slovak Research and Development
Agency grant APVV-0808-12 QIMABOS.  MS is supported by the NSF Grant
No. CCF-0832787, ``Understanding, Coping with, and Benefiting from,
Intractability'' and by CISE/MPS 1246641. UV was supported by ARO
Grant W911NF-12-1-0541, NSF Grant CCF-0905626, and Templeton Grants
21674 and 52536.

\end{document}